\documentclass[]{birkjour}
\usepackage[noadjust]{cite}
\usepackage{xcolor}
\RequirePackage[all]{xy}
\usepackage{graphicx} 
\usepackage{amsfonts}
\usepackage{amsmath}
\usepackage{amsthm}
\usepackage{physics}
\usepackage{cancel}
\usepackage{mathtools}
\usepackage{cleveref}

\newtheorem{thm}{Theorem}[section]
\newtheorem{cor}[thm]{Corollary}
\newtheorem{lem}[thm]{Lemma}

\theoremstyle{definition}
\newtheorem{defn}[thm]{Definition}
\theoremstyle{remark}
\newtheorem{rem}[thm]{Remark}

\newtheorem{ex}[thm]{Example}
\numberwithin{equation}{section}

\DeclareMathOperator{\co}{\mbox{\textsf{c}}}
\DeclareMathOperator{\si}{\mbox{\textsf{s}}}
\DeclareMathOperator{\ta}{\mbox{\textsf{t}}}
\DeclareMathOperator{\cota}{\mbox{\textsf{cot}}}

\newcommand{\BibTeX}{B\kern-0.1emi\kern-0.017emb\kern-0.15em\TeX}
\newcommand{\XYpic}{$\mathrm{X\kern-0.3em\raisebox{-0.18em}{Y}}$-$\mathrm{pic}\,$}

\newcommand{\cl}{C \kern -0.1em \ell}  

\newcommand{\ed}{\end{document}}

\begin{document}

\title[Outer and eigen: tangent concepts]
 {Outer and eigen: tangent concepts}
%
\author[Eelbode]{David Eelbode}
\email{david.eelbode@uantwerpen.be}
\author[Roelfs]{Martin Roelfs}
\email{martin.roelfs@uantwerpen.be}
\author[De Keninck]{Steven De Keninck}
\email{enkimute@gmail.com}
\date{\today}
\dedicatory{Last Revised:\\ \today}
\begin{abstract}
In this paper we use the power of the outer exponential $\Lambda^B$ of a bivector $B$ to see the so-called invariant decomposition from a different perspective. This is deeply connected with the eigenvalues for the adjoint action of $B$, a fact that allows a version of the Cayley-Hamilton theorem which factorises the classical theorem (both the matrix version and the geometric algebra version). 
\end{abstract}
\label{page:firstblob}
\maketitle
\tableofcontents

\begin{center}
    {\em Dedicated to the late professor Frank Sommen, \\
    the Gandalf meets Master Yoda of Clifford analysis.} 
\end{center}

\section{Introduction}
Apart from the generators of a geometric algebra, elements of grade 1 which can be interpreted as reflections in a plane-based view, one can argue that the bivectors (grade-2 elements which generate the even subalgebra) play an equally important role. Under the exponential map, these bivectors generate spin group elements (rotors $R$) and hence also have an undeniable interpretation in geometric terms (leading to rotations, for instance, in case of an Euclidean signature). 
In \cite{GSGPS}, Roelfs and De Keninck proved the `invariant decomposition' for bivectors, a generalisation of a conjecture by M. Riesz which states that a bivector in a geometric algebra $\mathbb{R}_{p,q,r}$ with $p+q+r=n$ can be decomposed into a sum of at most $\lfloor \frac{n}{2} \rfloor$ commuting bivectors (see \cite{Riesz}). The subscript $(p,q,r)$ hereby stands for the signature of the underlying (and possibly degenerate) orthogonal space of dimension $n$, where $r$ is the number of degenerate units (basis vectors $v_j$ for which $v_j^2 = 0$). 
The main result in \cite{GSGPS}, which in itself is a refinement of the classical Mozzi-Chasles Theorem, says that any $\ell$-reflection $U = u_1u_2\cdots u_\ell$ can be written as a product of exactly $\lceil \tfrac{\ell}{2} \rceil$ {\em commuting} factors. 
These are $\lfloor \tfrac{\ell}{2} \rfloor$ bireflections and, for odd $\ell$, one extra reflection. 
These factors are all {\em simple} and are thus generated by an element that squares to a scalar (for a bireflection this means that it is generated by a 2-blade, i.e. a product $B = vw$ of 2 orthogonal vectors). 
One can also see this result as a Lie algebra statement, because the commuting simple bivectors form a basis for the Cartan algebra for the orthogonal Lie algebra generated by the bivectors (see e.g. \cite{FH}).  \\
The aim of the present paper is to change the narrative slightly, to reinterpret the results from \cite{GSGPS} using the language of eigenvectors for the bivector $B$, and to explain why the so-called outer exponential $\Lambda^B$ plays an important role in this story. This function (defined in \cref{sec:outerexp}), not to be confused with the classical exponential $R = e^B$, already appeared in the work of Lounesto (see \cite{Lou}), but we will show how this particular function appears naturally in the framework of the eigenvalue equation for a bivector $B$. In particular, it leads to a factorisation (`a square root') of the classical matrix equation and its associated Cayley-Hamilton theorem (\cref{sec:ch}). The eigenvectors can then be combined into commuting simple bivectors which express the invariant decomposition (\cref{sec:eigenouter}). Moreover, this can all be expressed in terms of the so-called outer tangent function (defined in \cref{sec:outertrig}), which also appears due to the behaviour of the Cayley transform (\cref{sec:cayley}). Finally, we assume the reader is familiar with basic notions in a geometric (or Clifford) algebra such as grades and the properties of the multiplicative structure, if not we recommend classical sources such as \cite{HS, Lou, Po} or more pedagogical resources such as \cite{GA4Ph,GA4CS}. More familiarity with the (even more) geometric interpretation can be gained in e.g. the overview paper \cite{DKD}. 
\section{The outer exponential}\label{sec:outerexp}
Consider an arbitrary bivector $B \in \mathbb{R}_{p,q,r}^{(2)}$. The dimension of the subspace in which the bivector $B$ lives can be measured by computing the quantities
\begin{equation}
    W_j \coloneqq \frac{1}{j!} \overbrace{B \wedge B \wedge \cdots \wedge B}^{\text{j times}}\ = \frac {B^{{\scriptscriptstyle \wedge} j}}{j!}\ = \frac{1}{j!} \expval{ B^j }_{2j} ,
\end{equation}
whereby $B$ is wedged with itself $j$ times, and where $\expval{\cdots}_\ell$ is the $\ell$-grade selection operator. Indeed, the largest integer $k \in \mathbb{N}_0$ for which $W_k \neq 0$ determines the so-called {\em effective dimension} $2k$ for $B$. The $2k$-blade $W_k$ can then be seen as the pseudoscalar of the space in which the $B$-action takes place, and is referred to as the {\em effective pseudoscalar} (see \ref{W_k^2_real} for a proof that $W_k$ is indeed a pseudoscalar). Bivectors $B$ for which $W_k^2 = 0$ are special, because they always lead to a $2k$-blade $W_{k} = v_1 \wedge \cdots \wedge v_{2k}$ which contains at least one null generator $v_j$ (a reflection $v_j$ with $v_j^2 = 0$). To distinguish these bivectors from those for which $W_k^2 \neq 0$, we introduce the following concept: 
\begin{defn}
    A bivector $B \in \mathbb{R}_{p,q,r}^{(2)}$ is called \emph{pseudo-null} if its associated effective pseudoscalar $W_k$ is null, and hence $W_k^2 = 0$. A bivector which is not pseudo-null will be referred to as a {\em regular} bivector. 
\end{defn}
\begin{rem}
    In the special case where $B = W_1$ is null, the bivector itself is null but this is not true in general (hence the prefix `pseudo'). 
\end{rem}
\noindent
Note that most of the results that will be derived in this paper will hold for both regular and pseudo-null bivectors, but it will turn out to be handy to have a word to distinguish these. We shall see however, that the results of proofs which are only valid for the regular case will often be extendable to the pseudo-null case using a suitable limit involving the eigenvalues. Given a bivector $B$, it is well-known that $B$ has a well-defined grade preserving commutator action on the space of $\ell$-vectors, for $0 \leq \ell \leq n$, known as the adjoint representation:
\begin{align*}
    \textup{ad}(B) : \mathbb{R}_{p,q,r}^{(\ell)} \rightarrow\ \mathbb{R}_{p,q,r}^{(\ell)} : \alpha \mapsto \textup{ad}(B)[\alpha] \coloneqq \comm{B}{\alpha}\ = B\alpha - \alpha B.
\end{align*}
This notation is standard in the framework of Lie algebras, the appearance of which should not come as a surprise because the space of bivectors is indeed a model for the Lie algebra $\mathfrak{so}(p,q,r)$, a fact which lies behind one of the more influential papers by Sommen and others \cite{LGaSG}. However, it will not be the notation used in the present paper. Here, we will start from the fact that a product $AB$ can always be decomposed into a commutator and an anti-commutator:
\begin{equation*}
        AB = \frac{1}{2}[A,B] + \frac{1}{2}\{A,B\}\ .
\end{equation*}
Due to its special role, the commutator part is then traditionally defined as the commutator product:
\[ A \times B \coloneqq \frac{1}{2} [A,B]\ . \]
In Geometric Algebra (GA) one uses this product instead of the commutator, since it is often directly identifiable with one of the common products of GA. For example, for two vectors $u,v$ the commutator product is identical to the wedge product $u \times v = u \wedge v$, whereas for a bivector $B$ with a vector $u$ it equals the dot product $B \times u = B \cdot u$. 
A more geometric motivation for the commutator product stems from the fact that rotors are exponentials of bivectors. 
For example, in order to rotate an element $X$ in the plane by an angle $\theta$ around a point $p$ one forms the rotor $R = \exp(\tfrac{1}{2} \theta p)$ and transforms the element $X$ under the group action $X \to R X \widetilde{R}$.
At first order, this is equivalent to
    \[ X \to X + \theta p \times X, \]
where we used the expansion $R = 1 + \tfrac{1}{2} \theta p + \order{\theta^2}$.
The commutator product therefore also makes a natural appearance in transformations.
Furthermore, this establishes that the group action of $\textup{Spin}(p, q, r)$ and the adjoint action of $\mathfrak{spin}(p, q, r)$ are deeply linked, and both are \emph{grade preserving} operations. The commutator product is therefore more useful in GA than the classical commutator itself. It is then natural to wonder about the invariants of the commutator product with $B$, or equivalently the invariants under the group action. Because the commutator with $B$ is a linear map acting as a derivation $B \times uv = (B \times u)v + u(B \times v)$, the commutator with a composite element such as $uv$ can be understood fully by its action on vectors. This thus leads in a natural way to the study of eigenvectors $v_\mu$ under the commutator product with $B$, where $\mu$ will refer to the eigenvalue. 
\begin{defn}
    A vector $v_\mu \neq 0$ is an eigenvector for the bivector $B$ with eigenvalue $\mu \in \mathbb{C}$ if it satisfies $B\times v_\mu = B\cdot v_\mu = \mu v_\mu$. The set of eigenvalues for $B$, also known as its spectrum, will be denoted by $\sigma(B) \subset \mathbb{C}$. 
\end{defn}
\noindent
Rather than working with $B \times v_\mu = \mu v_\mu$ directly however, we will first show how one can recast this equation into an equation involving a single operator $M_\mu$ acting on $v$ from the left. 
This is motivated by our desire to work with the invertible geometric product, rather than the commutator product, as it enables us to define the spectrum of $B$ in terms of `singular' values $\mu$ (values for which the operator $M_\mu$ is not invertible). In order to derive the operator $M_\mu$, we need the following powerful lemma.
\begin{lem}\label{lem:W_recursive}
Suppose $B \in \mathbb{R}_{p,q,r}^{(2)}$ and $v \in \mathbb{R}_{p,q,r}^{(1)}$. If the image of the vector $v$ under the action of $B$ is denoted by means of $w = B \times v$, one has that
\begin{align}
W_j \cdot v &= W_{j-1} \wedge w = W_{j-1} \wedge (B \times v)
\end{align}
for all $j$, where we define $W_0 \coloneqq 1$. 
\end{lem}
\begin{proof}
    The proof for this statement relies on the simple observation that $W_j \cdot v$ is an element of grade $(2j-1)$. Since $W_j \times B = 0$, we get that 
    \[ W_j B = W_j \cdot B + \cancel{W_j \times B} + W_j \wedge B = W_j \cdot B + (j+1 )W_{j+1}\ . \]
    This means that $W_{j+1} = \tfrac{1}{j+1} \big( W_j B - W_j \cdot B\big)$. The upshot is that the term $W_j \cdot B$ has grade $(2j-2)$, which means that it cannot contribute to $W_{j+1} \cdot v$, as this is an element of grade $(2j+1)$. In other words, we have that 
    \begin{align*}
        W_{j+1} \cdot v &= \frac{1}{j+1} (W_j B) \cdot v = \frac{1}{(j+1)!}\expval{B^{j+1} \cdot v}_{2j+1} \\
        &= \frac{1}{(j+1)!}\expval{(B \cdot v) B^j + B (B \cdot v) B^{j-1} + \ldots + B^j (B \cdot v)}_{2j+1} \\
        &= \frac{1}{(j+1)!}\expval{w B^j + B w B^{j-1} + \ldots + B^j w }_{2j+1}\\ 
        &= \frac{1}{j!}\expval{B^jw}_{2j+1} = W_{j} \wedge w\ .
    \end{align*}
    Here we used the relation $wB = Bw - 2B\cdot w$ to commute $w$ to the right-hand side in the penultimate equality, thereby ignoring the dot products (which is allowed in view of the grade selection operator). 
\end{proof}
\noindent
\noindent
The previous lemma has a few consequences, which we will now investigate. First of all, we can now show that $W_k \neq 0$ is indeed an effective {\em pseudoscalar} in the space defined by a bivector $B$ (see the introduction). 
\begin{lem}\label{W_k^2_real}
    For a bivector $B \in \mathbb{R}_{p,q,r}^{(2)}$, the effective pseudoscalar $W_k$ satisfies the following properties:
    \begin{enumerate}
        \item For any vector $w = B \times v$ one has that $W_k \wedge w = 0$.
        \item The square of the pseudo-scalar is real, i.e. $W_k^2 \in \mathbb{R}$.
    \end{enumerate}
\end{lem}
\begin{proof}
    The proof to the first statement is a corollary of \ref{lem:W_recursive}:
    \[ W_k \wedge w = W_k \wedge (B\cdot v) = W_{k+1} \cdot v = 0. \] 
    Defining the vector $w$ as $B \times v$ ensures that the vector $w$ lies in the subspace of $B$ and thus also of $W_k$.
    In order to prove that $W_k^2 \in \mathbb{R}$, it suffices to show that $W_k^2$ is a central, i.e. commuting, element in the Clifford algebra: because $W_k^2$ has even degree, it can then only be a scalar. 
    Take an arbitrary vector $v_0$. 
    Either $B \cdot v_0 = 0$, in which case it trivially follows that also $W_k^2$ commutes with $v_0$, or $B \cdot v_0 = v_1 \neq 0$. 
    In the latter case we have that 
    \[ W_k^2\times v_0 = W_k(W_k \cdot v_0) + (W_k \cdot v_0)W_k = \{W_k,W_{k-1} \wedge v_1\}\ , \]
    where we have used the previous lemma. Since $W_k \wedge v_1 = 0$, and using the fact that the elements $W_i$ and $W_j$ always commute (for all $i$ and $j$), we have that
    \begin{align*}
       \{W_k,W_{k-1} \wedge v_1\} &= \{W_k,W_{k-1}v_1 - W_{k-1}\cdot v_1\}\\ 
       &= W_{k-1}(W_k \wedge v_1) - \{W_k,W_{k-1}\cdot v_1\}\\ 
       &= - \{W_k,W_{k-1}\cdot v_1\}\ . 
    \end{align*}
     If $B\cdot v_1 = 0$, then also $W_{k-1}\cdot v_1 = 0$ and the result is proved. If not, we know that $W_{k-1}\cdot v_1 = W_{k-2} \wedge v_2$ with $B \cdot v_1 = v_2$. But since 
    \[ \{W_k,W_{k-2} \wedge v_2\} = \{W_k,W_{k-2}v_2 - W_{k-2}\cdot v_2\} = -\{W_k,W_{k-2}\cdot v_2\}\ ,  \]
    hereby using that $\{W_k,W_{k-2}v_2\} = W_{k-2}(W_k\wedge v_2) = 0$, we have reduced the proof to an induction argument on the parameter $k$. After $(k-1)$ steps we arrive at the expression $\{W_k,W_1\cdot v_{k-1}\}$, which is always zero because the outer product of $W_k$ with a vector in the image of $B$ is trivial. 
\end{proof}
\begin{rem}
    The precise meaning of the scalar $W_k^2$ will become clear at the end of this section, see \cref{cor_prod_eigenvals}. 
\end{rem}
\noindent
Suppose that we now have an eigenvector $v_\mu \neq 0$ for a regular bivector $B$ with eigenvalue $\mu \in \mathbb{C}_0$. Note that zero eigenvalues $\mu = 0$ are excluded at this point, because we will concoct an argument which involves being able to divide by $\mu$. However, once we have reached our final conclusion, we will be able to extend it in such a way that no restrictions on the eigenvalues have to be imposed. Lemma \ref{lem:W_recursive} tells us that $W_j\cdot v_\mu = \mu W_{j - 1} \wedge v_\mu$. Repeatedly rewriting the wedge product in terms of the geometric and dot products, we find: 
\begin{align*}
        W_k v_\mu &= \mu W_{k-1} \wedge v_\mu\ =\ \mu W_{k-1} v_\mu - \mu W_{k-1} \cdot v_\mu \\
        &= \mu W_{k-1} v_\mu - \mu^2 W_{k-2} \wedge v_\mu = \cdots\\
        &= \big(\mu W_{k-1} - \mu^2 W_{k-2} + \ldots + (-1)^{k-1} \mu^k\big) v_\mu\ .
\end{align*}
This calculation suggests looking at the equation 
\[ M_\mu v_\mu \coloneqq \big(W_k - \mu W_{k-1} + \mu^2 W_{k-2} + \ldots + (-1)^k \mu^k\big)v_\mu = 0\ . \]
For $\mu \neq 0$, this equation is actually {\em equivalent} to the original eigenvalue equation $B\times v_\mu = \mu v_\mu$. This follows from the projection of the equation $M_\mu v_\mu = 0$ on the 1-graded part. Since $\mu \neq 0$, we can rewrite this eigenvalue equation as follows: 
\[ M_\mu v_\mu = (-\mu)^k\left(1 - \frac{B}{\mu} + \frac{W_2}{\mu^2} - \ldots + (-1)^k \frac{W_k}{\mu^k}\right)v_\mu = 0\ , \]
whereby the sum between brackets can be recognised as {\em the outer exponential} of the bivector $- B / \mu$, defined below: 
\begin{defn}
The outer exponential $\Lambda^B$ of a bivector $B \in \mathbb{R}_{p,q,r}^{(2)}$ is given by
\begin{equation}
\Lambda^B \coloneqq 1 + B + \frac{1}{2!}B \wedge B + \ldots = \sum_{j  = 0}^\infty  \frac{B^{\scriptscriptstyle \wedge j}}{j!} = \sum_{j  = 0}^k W_j\ .
\end{equation}
Note that the outer exponential will always be a finite sum, ending with the term for $j = k$ (the effective pseudoscalar is the last term in the summation). This stands in sharp contrast with the classical exponential, defined using an infinite series. We will discuss a geometrical interpretation of the outer exponential in \cref{sec:outergeo}
\end{defn}
\noindent
So, given a bivector $B \in \mathbb{R}_{p,q,r}^{(2)}$ and an eigenvalue $\mu \in \mathbb{C}_0$, the equation for the eigenvectors can be rewritten as 
\begin{align}\label{eigenvalue_outer_exp}
    \Lambda^{-\frac{B}{\mu}}v_\mu &= 0\ .
\end{align}
This suggests that the eigenvalues are precisely those numbers $\mu \in \mathbb{C}$ for which the outer exponential in the formula above is non-invertible. As a matter of fact, one has the following property (which also appears in \cite{Lou}, but we decided to include it here with an explicit proof): 
\begin{lem}\label{lem:mfs}
For any bivector $B \in \mathbb{R}_{p,q,r}^{(2)}$, the quantity $\vert \Lambda^{B}\vert ^2 = \Lambda^{B} \Lambda^{-B}$ is always a real-valued scalar. Note also that $|\Lambda^B|^2 = |\Lambda^{-B}|^2$. 
\end{lem}
\begin{proof}
Because $\Lambda^{B} \Lambda^{-B}$ is a self-reverse element of the even subalgebra, it can only contain elements of grade $4j$ with $j \in  \mathbb{N}$. To prove that $\Lambda^{B} \Lambda^{-B}$ is scalar, it therefore suffices to show that it commutes with any vector $v$ and hence defines a central element. Putting $B \times v = B \cdot v = w$, we then have that 
    \begin{align*}
        (\Lambda^{B} \Lambda^{-B}) \times v &= (\Lambda^{B} \times v) \Lambda^{-B} + \Lambda^{B} (\Lambda^{-B} \times v) \\ 
        &= (\Lambda^{B} \wedge w) \Lambda^{-B} - \Lambda^{B} (\Lambda^{-B} \wedge w)\ ,
    \end{align*}
where we used the fact that for all $\theta \in \mathbb{R}$ one has
\begin{align*}
    \Lambda^{\theta B} \times v &= \sum_{j = 1}^k\: \theta ^j W_j \times v = \sum_{j = 1}^k \theta ^j W_{j-1} \wedge w = \theta \, \Lambda^{\theta B} \wedge w\ .
\end{align*}
Invoking the relation $\{A,B\}C - A\{B,C\} = [B,AC] = -[AC,B]$, and using that $\Lambda^B \wedge w = \tfrac{1}{2}\acomm{\Lambda^B}{w}$ is an anti-commutator, we find
\begin{align*}
    (\Lambda^{B} \wedge w) \Lambda^{-B} - \Lambda^{B} (\Lambda^{-B} \wedge w) = - (\Lambda^{B} \Lambda^{-B}) \times w\ .
\end{align*}
Using a similar reasoning, we find that $(\Lambda^{-B}\Lambda^{B}) \times v = +(\Lambda^{-B}\Lambda^{B}) \times w$, where the relative minus sign is crucial. Since $\Lambda^{B}$ and $\Lambda^{-B}$ commute, this leads to
\[ - (\Lambda^{B}\Lambda^{-B}) \times w = (\Lambda^{B}\Lambda^{-B}) \times v = (\Lambda^{-B}\Lambda^{B}) \times v = +(\Lambda^{-B}\Lambda^{B}) \times w\ . \]
This says that $(\Lambda^{B}\Lambda^{-B}) \times v$ is equal to plus {\em and} minus the same expression, and hence trivial, which concludes the proof. 
\end{proof}
\noindent
The equation $M_\mu v = 0$ can only have non-trivial solutions when $M_\mu$ is not invertible.
In view of the fact that $M_\mu$ can be written as an outer exponential, lemma then \ref{lem:mfs} implies that eigenvectors for $B$ are associated to $\mu$ for which $M_\mu \widetilde{M}_\mu = 0$.
\begin{defn}
    The spectrum $\sigma(B)$ of a bivector $B \in \mathbb{R}^{(2)}_{p,q,r}$ is defined as the solutions (over $\mathbb{C}$) of the scalar equation 
    \begin{align} \label{master_polynomial}
        P_{2k}(\mu) \coloneqq M_\mu \widetilde{M}_\mu = \mu^{2k}\Lambda^{-\frac{B}{\mu}}\Lambda^{+\frac{B}{\mu}} = 0\ .
    \end{align}
The integer $2k \in \mathbb{N}_0$ is the previously defined effective dimension of $B$. 
\end{defn}
\begin{rem}
    Note that $P_{2k}(\mu)$ is still defined if $\mu = 0$ is an eigenvalue, despite the fact that the outer exponentials will not be defined. Indeed, the factor $\mu^{2k}$ cancels the apparent poles. 
\end{rem}
\noindent
Note that $P_{2k}(\mu)$ is a polynomial of degree $2k$ in the variable $\mu \in \mathbb{C}$, with real coefficients. It is also clear that $P_{2k}(\mu)$ is even in $\mu$, which implies that solutions appear in pairs $\pm \mu_i \in \mathbb{C}$ (this even holds for zero eigenvalues). Put differently, we have that $P_{2k}(\mu) = Q_k(\lambda)$, where $Q_k$ is a polynomial of degree $k$ and $\lambda = \mu^2$. As a result, there is a connection between the eigenvalues $\pm\mu_i$ (for $1 \leq i \leq k$) and the quantities $W_j$ associated to $B$. To see this, we will need the so-called elementary symmetric polynomials of degree $p$ in $d$ variables, given by
\[ e_p(x_1,\ldots,x_d) \coloneqq \sum_{1 \leq i_1 < \ldots < i_p \leq d}x_{i_1}\ldots x_{i_p}\ , \]
where it is tacitly assumed that $1 \leq p \leq d$ (one can allow the index $p = 0$ if one defines $e_0$ as the constant function 1). 
\begin{lem}\label{W_j^2_sympol}
    For a bivector $B \in \mathbb{R}_{p,q,r}^{(2)}$, one has that $\expval{W_j^2}_0 = e_j(\mu_1^2,\ldots,\mu_k^2)$, and this for all indices $1 \leq j \leq k$ where $2k$ is the effective dimension. 
\end{lem}
\begin{proof}
    On the one hand, the characteristic equation for $B$ is given by 
    \begin{align*}
        0 = P_{2k}(\mu) &= \big(\mu^k - \mu^{k-1}B + \ldots + (-1)^k W_k\big)\big(\mu^k + \mu^{k-1}B + \ldots + W_k\big)\ ,
    \end{align*}
    where we once again stress that this is a {\em scalar} equation by lemma \ref{lem:mfs}. On the other hand, given the solutions $\pm \mu_i$ (with $1 \leq i \leq k$), it is well-known that the characteristic equation can be expressed in terms of these eigenvalues as
    \[ P_{2k}(\mu) = \prod_{j=1}^k (\mu^2 - \mu^2_j) = \sum_{j = 0}^{2k} (-1)^j e_{j}(\mu_1,-\mu_1,\ldots,\mu_k,-\mu_k)\mu^{2k-j}\ . \]
    However, the odd symmetric polynomials will be trivial in the summation above, precisely because the eigenvalues appear in pairs $\pm \mu_i \in \mathbb{C}$. Even more, for even indices we get that $e_{2j}(\mu_1,-\mu_1,\ldots,\mu_k,-\mu_k) = (-1)^j e_{j}(\mu_1^2,\ldots,\mu_k^2)$, which means that 
    \[ \sum_{j = 0}^k (-1)^j e_{j}(\mu_1^2,\ldots,\mu_k^2) \mu^{2k-2j} = \big(\mu^k \Lambda^{-\frac{B}{\mu}}\big)\big(\mu^k \Lambda^{+\frac{B}{\mu}}\big)\ .  \]
    To arrive at the conclusion of the lemma, it is therefore sufficient to compare the coefficients of $\mu^{2k-2j}$ at both sides of the equation (hereby taking into account that this equation is real-valued). Now, when looking at $\mu^{2k-2j}$, it is clear that this leads to an equation of the form $e_{j}(\mu_1^2,\ldots,\mu_k^2) = W_j^2 + \mbox{rest}$, whereby `rest' stands for products of the form $W_aW_b$ whereby $a \neq b$ (note that `rest' can be zero). These products are then necessary to ensure that the non-scalar parts of $W_j^2$ disappear, but they never contribute to the scalar part itself. Taking the scalar part of the last equality above then proves the statement. 
\end{proof}
\noindent
The case $j = k$ is interesting in its own right, as it gives us a connection between regular/pseudo-null bivectors and their spectrum:  
\begin{cor}\label{cor_prod_eigenvals}
    If $W_k$ is the effective pseudoscalar associated to a bivector $B$, then $W_k^2$ is the product of all the squared eigenvalues for $B$.
\end{cor} 

\subsection{Outer Exponential Geometry}\label{sec:outergeo}
While the ordinary exponential, with its infinite Taylor expansion, has a continuous feel to it, the outer exponential as a sum of at most $k$ discrete terms is in effect a sum of just a few discrete transformations. To get a geometric intuition for how the outer exponential still generates (unnormalized) $\textup{Spin}$ transformations,
consider as an example a simple bireflection $R = a + bB$ encoding a rotation in the Euclidean plane, where $a,b \in \mathbb R, a^2 + b^2 = 1$ and $B$ is a normalized bivector representing the point we are rotating around.

We aim to understand how all such rotations can be written as a linear combination of the identity bireflection and a point reflection in the center of rotation $B$.
From this perspective, the weight $a$ tells us how important the identity contribution is, while $b$ tells us the importance of our point-reflected contribution. Since both these contributions lie on the same line, it seems strange that this could ever produce rotations.
However, the paradox is resolved when we work out the transformation of, e.g. a point $P$ under the bireflection $R$:
\[P' = R P \widetilde{R} = (a + bB) P (a - bB) = a^2P + 2ab(B \times P) - b^2 BPB \ . \]
The final rewrite now shows us that the transformed point $P'$ is, in fact, the linear combination of \emph{three} terms.
As shown in \cref{fig:outergeo}, with a weight of $a^2$ we have the identity contribution (blue), the reflected contribution gets a weight of $b^2$ (pink), and the remaining contribution is given by the commutator product (black). 
The geometry of this commutator is the key to the discrete interpretation of a rotation: since both $P$ and $B$ are points in our example, their product is a translation.
However, $PB$ and $BP$ are translations in opposite directions, making their difference a point at infinity, orthogonal to the line between $P$ and $B$. Adding this point at infinity will therefore make the resulting point move off the line $P \vee B$.
\begin{figure}
    \centering
    \includegraphics[width=1\linewidth]{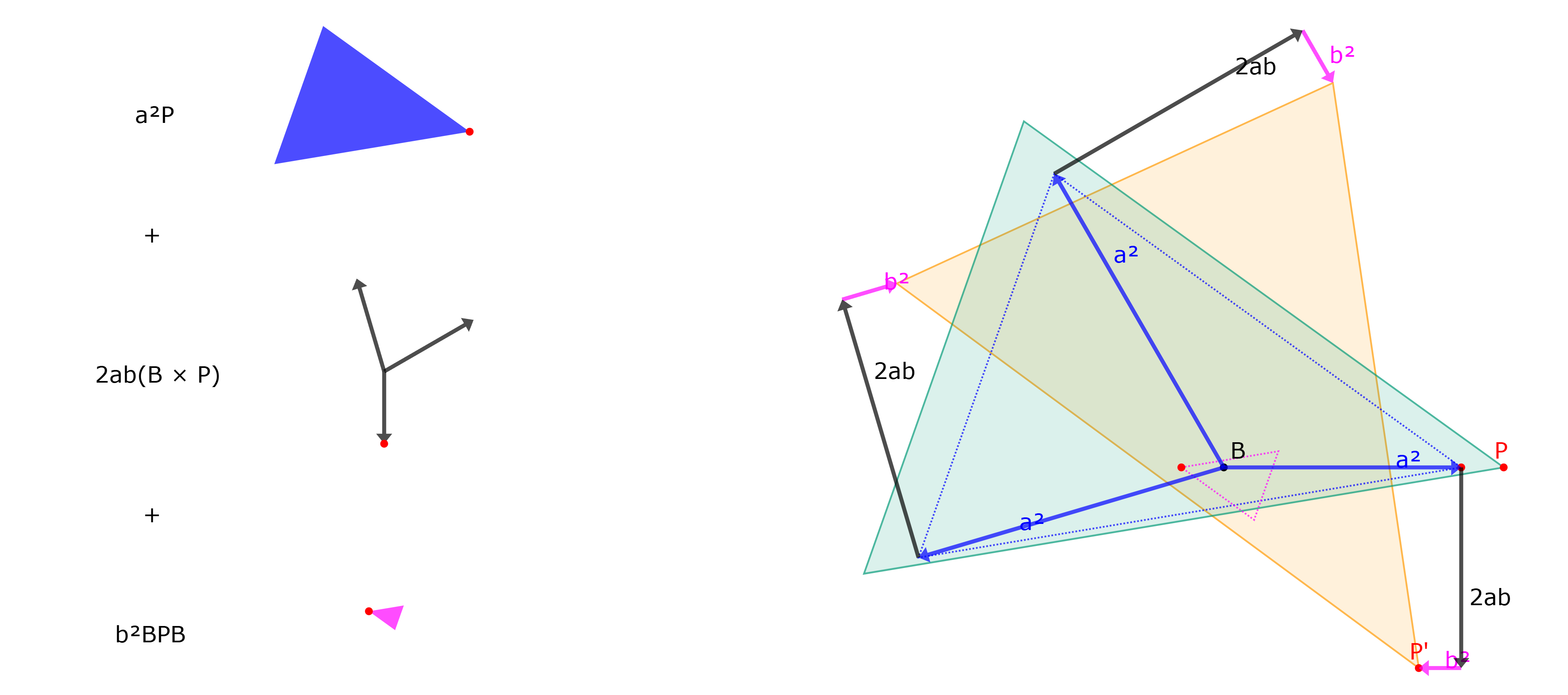}
    \caption{The rotation of the green triangle into the orange one as a weighted sum of the untransformed (blue), point reflected (pink) and commutator (black) contributions.}
    \label{fig:outergeo}
\end{figure}

The situation is similar for other grades of $P$, since the identity, reflected, and commutator product contributions are all grade preserving.
Moreover, the same mindset works in 3D.
For example, in 3DPGA, given $B$ a bivector line and $P$ a trivector point, their compositions $BP$ and $PB$ are both tri-reflections (transflections) that share the same reflection and have opposite translation parts. 
Their difference will again be an ideal point orthogonal to the plane $B \vee P$, just as it is in the 2D case.

Hence, after normalization the outer exponential's natural interpretation as a linear combination of discrete transformations offers an alternative way to think about transformations.
\section{Outer trigonometry}\label{sec:outertrig}
Consider a bivector $B \in \mathbb{R}_{p,q,r}^{(2)}$ with $W_k \neq 0$ the effective pseudoscalar. We will then show how the outer exponential can be used to say something about the decomposition of $B$ into commuting simple bivectors. To do so, we need the outer trig functions. In view of the natural connection between the exponential function and the trigonometric functions (be it elliptic or hyperbolic), it seems natural to define the outer trigonometric functions as follows: 
\begin{defn}
    For $B \in \mathbb{R}_{p,q,r}^{(2)}$, the outer trigonometric functions are given by: 
    \begin{align*}
        \co_{\wedge}(B) = \tfrac{1}{2}(\Lambda^B + \Lambda^{-B}) \qquad \si_{\wedge}(B) = \tfrac{1}{2}(\Lambda^B - \Lambda^{-B}) \qquad \ta_{\wedge}(B) = \frac{\si_\wedge(B)}{\co_\wedge(B)}
    \end{align*}
    The last expression can obviously be ill-defined, when $\co_\wedge(B)$ is not invertible. 
\end{defn}
\noindent
These outer trigonometric functions share many of the classical properties. In order to prove these properties, we will use the fact that the elements $W_j$ satisfy the following relation:
\begin{align}\label{wedge_W}
    W_i \wedge W_j &= \frac{1}{i!j!}\expval{B^i}_{2i}\wedge\expval{B^{j}}_{2j} = {i + j \choose j}\frac{\expval{B^{i + j}}_{2(i + j)}}{(i+j)!} = {i + j \choose j}W_{i+j}
\end{align}
with $0 \leq i, j \leq k$ and where $W_{i+j}$ is trivial whenever $i+j > k$. Also note that when $i \neq j$, one has that the scalar part $\expval{W_iW_j}_0 = 0$. The first thing to look at is the main identity, but here it matters whether the Clifford product or the outer product is used. Since $\Lambda^B \wedge \Lambda^{-B} = 1$, which follows from straight-forward calculations based on formula (\ref{wedge_W}), one gets that 
\begin{align*}
    \big(\co_{\wedge}(B)\big)^{\wedge 2} - \big(\si_{\wedge}(B)\big)^{\wedge 2} = 1\ ,
\end{align*}
where $\alpha^{\wedge k}$ stands for the $k$-fold outer product of an element $\alpha$ with itself, and 
\begin{align*}
    \big(\co_{\wedge}(B)\big)^{ 2} - \big(\si_{\wedge}(B)\big)^{ 2} = \big\vert\Lambda^B\big\vert^2
\end{align*}
Note that the signs appearing here are reminiscent of the hyperbolic functions rather than the trigonometric functions. To derive relations which mimic the classical sum and difference relations for sine and cosine, the following relation will come in handy: 
\begin{lem}
If $B_1$ and $B_2 \in \mathbb{R}_{p,q,r}^{(2)}$ are commuting bivectors, i.e. $B_1 \times B_2 = 0$, then one has that
\[ W^{(12)}_a = \sum_{j = 0}^a W_j^{(1)}\wedge W_{a-j}^{(2)}\ , \]
where the upper indices $(1)$ and $(2)$ refer to the bivectors $B_1$ and $B_2$, and where $(12)$ refers to the bivector sum $B = B_1 + B_2$. 
\end{lem}
\begin{proof}
In view of the fact that the bivectors commute, one has that 
\begin{align*}
    W^{(12)}_a &= \frac{1}{a!}\expval{(B_1 + B_2)^a}_{2a} = \frac{1}{a!}\expval{\sum_{j = 0}^a {a \choose j}B_1^j B_2^{a - j}}_{2a}\\ 
    &= \sum_{j = 0}^a \frac{1}{j!(j-a)!}\expval{B_1^j}_{2j} \wedge \expval{B_2^{a-j}}_{2(a-j)}\ ,
\end{align*}
which reduces to the desired result. 
\end{proof}
\begin{cor}
    If two bivectors $B_1$ and $B_2$ commute, i.e. $B_1\times B_2 = 0$, then 
    \[ \Lambda^{B_1 + B_2} = \Lambda^{B_1} \wedge \Lambda^{B_2}\ . \]
\end{cor}
\noindent
Note that the previous result is expressed in terms of outer products. This is interesting in its own right, but it is often better if one can make use of the Clifford product. A slightly stronger set of requirements on the bivectors $B_1$ and $B_2$ leads to the following: 
\begin{lem}\label{QuadB_product}
If $B_1$ and $B_2 \in \mathbb{R}_{p,q,r}^{(2)}$ are commuting bivectors with $\expval{B_1B_2}_0 = 0$, which thus means that $B_1B_2 = B_1 \wedge B_2$, then one has that
\begin{equation}\label{comm_product}
    \Lambda^{B_1 + B_2} = \Lambda^{B_1} \Lambda^{B_2}\ .
\end{equation}
\end{lem}
\begin{proof}
Because $B_1B_2 = B_1 \wedge B_2$, one may conclude that 
\begin{align*}
    W^{(12)}_a &= \frac{1}{a!}\expval{(B_1 + B_2)^a}_{2a} = \frac{1}{a!}\expval{\sum_{j = 0}^a {a \choose j}B_1^j B_2^{a - j}}_{2a}\\ 
    &= \sum_{j = 0}^a \frac{1}{j!(j-a)!}\expval{B_1^j}_{2j}  \expval{B_2^{a-j}}_{2(a-j)}\ ,
\end{align*}
where the wedge product has been replaced with the standard Clifford product. Indeed, since each of the products at the right-hand side of the sum is already an element of grade $2a$ (for all $j$), the outer product sign may be ignored.
\end{proof}
\begin{rem}
Note that the extra condition $B_1 B_2 = B_1 \wedge B_2$ ensures that the (commuting) bivectors $B_1$ and $B_2$ are linearly independent as elements of the GA. This extra condition is necessary though, which can easily be seen by considering the case $B_2 = -B_1$ in formula (\ref{comm_product}) above: the left-hand side equals 1 (by definition), but the right-hand side is equal to $|\Lambda^{B_1}|^2$ and this is not necessarily equal to 1. 
\end{rem}
\noindent
Using the relations obtained above, one can then derive formulas such as 
\begin{align}
    \co_\wedge(B_1 + B_2) &= \co_\wedge(B_1) \wedge \co_\wedge(B_2) +  \si_\wedge(B_1) \wedge \si_\wedge(B_2) \label{out_cos_sum}\\
    \si_\wedge(B_1 + B_2) &= \co_\wedge(B_1) \wedge \si_\wedge(B_2) +  \si_\wedge(B_1) \wedge \co_\wedge(B_2)\\
    \ta_\wedge(B_1 + B_2) &= \frac{\ta_\wedge(B_1) + \ta_\wedge(B_2)}{1 + \ta_\wedge(B_1)\wedge\ta_\wedge(B_2)}\label{out_tan_sum}
\end{align}
and the (many) likes thereof, provided the bivectors in the argument of these functions commute. If, moreover, also $\expval{B_1B_2}_0 = 0$, then wedge product may even be replaced by the Clifford product. This does suggest that being able to decompose a bivector into commuting simple bivectors makes outer trig identities easier (even more so if their product is a quadvector). But as we will soon observe, the opposite statement also holds: one can use the outer trig functions to say something about the (invariant) decomposition of a bivector.\\
\noindent
Let us consider a regular bivector $B$ with $W_k^2 \in \mathbb{R}_0$. Here we will again exclude pseudo-null bivectors (having $\mu = 0$ in their spectrum) at first, which safely allows us to divide by $\mu$, but we will include them at the end of our argument. We thus know that the spectrum of $B$ consists of pairs of non-trivial eigenvalues $\pm \mu_j$, with $1 \leq j \leq k$. For any such pair we can pick one of the (opposite) eigenvalues $\mu$, and use this to define an associated bivector 
\[ B_\mu \coloneqq \frac{1}{\mu}B \in \mathbb{R}_{p,q,r}^{(2)}\ . \]
Note that it does not really matter whether one uses $\mu$ or $-\mu$ (see later). The bivectors will obviously be different (with $B_{-\mu} = -B_\mu$), but the analysis below will still hold, regardless of the sign. However, we must keep in mind that in order to use this analysis to conclude something about the bivector $B$ we started from, we will need to take this scaling factor into account. The upshot is that in switching from $B$ to $B_\mu$, one obtains a bivector for which $|\Lambda^{B_\mu}|^2 = 0$, which means that $\co_\wedge^2(B_\mu) = \si_\wedge^2(B_\mu)$. All in all, this inspires us to have a closer look at bivectors $B$ for which the outer tangent satisifes $\ta_\wedge^2(B) = 1$, provided $\ta_\wedge(B)$ exists. 
\begin{lem}\label{pm1_spectrum}
For a bivector $B \in \mathbb{R}_{p,q,r}^{(2)}$, the following is true: 
\[ \mu = \pm 1 \in \textup{spec}(B) \Longleftrightarrow \co_\wedge^2(B) = \si_\wedge^2(B)\ . \]
If $\ta_\wedge(B)$ exists, this is equivalent with saying that $\ta_\wedge^2(B) = 1$. 
\end{lem}
\begin{proof}
    Both implications follow from the fact that 
    \begin{align*}
       4\si_\wedge^2(B) &= (\Lambda^B)^2 + (\Lambda^{-B})^2 - 2\Lambda^{-B}\Lambda^B \\
       4\co_\wedge^2(B) &= (\Lambda^B)^2 + (\Lambda^{-B})^2 + 2\Lambda^{-B}\Lambda^B\ .
    \end{align*}
    If $\mu = \pm 1$ is an eigenvalue, then $\Lambda^{-B}\Lambda^B = 0$ which implies $\co_\wedge^2(B) = \si_\wedge^2(B)$. Vice versa, if this equality holds, one finds that $\Lambda^{-B}\Lambda^B = 0$ but this is equivalent with saying that $\mu = \pm 1$ is an eigenvalue for $B$. 
\end{proof}
\begin{rem}
    One might think that the values $\mu = \pm 1$ are crucial here, but the upshot is that upon division by $\mu \neq 0$ one can always force these eigenvalues into the spectrum of $B_\mu$. From now on, we will therefore use $B_\mu$ with $\mu \neq 0$ as a notation for general bivectors with eigenvalue $\mu = \pm 1$. This does mean that the intermediate bivector $B_\mu$ is {\em not} always a {\em real} bivector, but the final result will always be real (after multiplication with $\mu$). 
\end{rem}
\noindent
The following lemma shows that one can `perturb' a given bivector, using a {\em commuting} and {\em linearly independent} bivector, without any impact on the spectrum. Note that the condition on being linearly independent is natural here: otherwise one could just add the opposite bivector, which obviously has a massive impact on the spectrum (one then gets only zero eigenvalues as a result). 
\begin{lem}
If $B_1$ is a bivector such that $\co_\wedge^2(B_1) = \si_\wedge^2(B_1)$ and $B_2$ is another bivector such that $B_1B_2 = B_1\wedge B_2$, then also $\co_\wedge^2(B_1 + B_2) = \si_\wedge^2(B_1+B_2)$.
\end{lem}
\begin{proof}
This follows from lemmas \ref{QuadB_product} and \ref{pm1_spectrum}, observing that
\[ \Lambda^{B_1 + B_2}\Lambda^{-(B_1+B_2)} = \Lambda^{B_1}\Lambda^{-B_1}\Lambda^{B_2}\Lambda^{-B_2} = 0\ , \]
hereby using the fact that $\vert\Lambda^{B_1}\vert^2 = 0$. 
\end{proof}
\noindent
In what follows, we will express $\Lambda^B$ as a product of (commuting) factors $\Lambda^{b_j}$ with $b_j$ a simple bireflection. The reason for doing so is the following lemma, which shows how to use the outer exponential to say something about bivectors: 
\begin{lem}\label{out_to_bivector}
    Suppose that $B_1, B_2, B \in \mathbb{R}_{p,q,r}^{(2)}$, where the bivectors $B_1$ and $B_2$ commuting and linearly independent (so $B_1B_2 = B_1 \wedge B_2$). One can then conclude that $\Lambda^B = \Lambda^{B_1}\Lambda^{B_2}\ \Rightarrow\ B = B_1 + B_2$. 
\end{lem}
\begin{proof}
    First of all, it is clear that if two outer exponentials are equal, then their bivector exponents are also equal (it suffices to consider the projection on the subspace of $2$-graded elements). To see why the lemma is true, it then suffices to note that 
    \[ \expval{\left(1 + B_1 + \frac{1}{2}B_1 \wedge B_1 + \ldots\right)\left(1 + B_2 + \frac{1}{2}B_2 \wedge B_2 + \ldots\right)}_2 = B_1 + B_2\ . \]
    This step explicitly uses the fact that $\expval{B_1B_2}_2 = 0$. 
\end{proof}
\noindent
Let us again consider a bivector $B$ with $\mu \neq 0$ an eigenvalue. We will often make use of the outer tangent function below, which means that in what follows we tacitly assume that $\co_\wedge(B_\mu)$ is invertible (we will come back to this assumption later). Let us first of all note the following: 
\begin{align*}
    \ta_\wedge^2(B_\mu) = 1\ &\Rightarrow\ \big(\co_\wedge(B_\mu) + \si_\wedge(B_\mu)\big)\big(\co_\wedge(B_\mu) - \si_\wedge(B_\mu)\big) = 0\\
    &\Leftrightarrow\ \Lambda^{B_\mu}\co_\wedge(B_\mu) = \Lambda^{B_\mu}\si_\wedge(B_\mu)\ ,
\end{align*}
which gives rise to the relation $\Lambda^{B_\mu} = \Lambda^{B_\mu} \ta_\wedge(B_\mu)$ if the outer tangent exists. 
\begin{ex}\label{si=Bco}
    Note that $\co_\wedge(B_\mu)$ not being invertible can also occur when $B$ is regular. Consider for instance the bivector $B = e_{12} + e_{34} \in \mathbb{R}_{4,0,0}$. It is clear that $\sigma(B) = \{\pm i\}$, where each eigenvalue occurs with algebraic multiplicity two (which makes $B$ the generator of a so-called isoclinic rotation). Dividing $B$ by $\mu = i$, we thus have that $1 \in \sigma(-iB)$, but $\co_\wedge(-iB) = 1 - e_{1234}$ is not invertible. Indeed, since $e_{1234}^2 = W_2^2 = 1$, we have that $(1 - W_2)(1 + W_2) = 0$, so we are dealing with a zero-divisor. However, it is easily verified here that $\si_\wedge(e_{12}) = e_{12}\co_\wedge(e_{12})$, and similarly for $e_{34}$, which means that $\ta_\wedge(e_{12})$ and $\ta_\wedge(e_{34})$ are still defined (as the problematic factor can be canceled out). As a matter of fact, this is a key observation, deeply connected with the existence of eigenvectors.
\end{ex}
\noindent
As will become clear in the next section, when we include the eigenvectors for $B$ into the story, the anti-self-reversed element $\ta_\wedge(B_\mu)$ is always a $2$-blade (whenever it exists). The bivectors $B_\mu$ and $\ta_\wedge(B_\mu)$ obviously commute (for $\mu \neq 0$), but they even lead to linearly independent bivectors. For that purpose, let us write
\[ B_\mu = \ta_\wedge(B_\mu) + \big(B_\mu - \ta_\wedge(B_\mu)\big) \coloneqq \ta_\wedge(B_\mu) + B_r\ , \]
where the subscript `r' stands for `the rest' of the bivector $B_\mu$. We then indeed have that $\expval{B_r \ta_\wedge(B_\mu)}_0 = 0$, which means that $B_r\ta_\wedge(B_\mu) = B_r \wedge \ta_\wedge(B_\mu)$. This easily follows from the previous lemma, which says that
\[ \expval{B_r\ta_\wedge(B_\mu)}_0 = \expval{(B_\mu - \ta_\wedge(B_\mu))\ta_\wedge(B_\mu)}_0 = 0\ .\]
One can then also prove the following: 
\begin{thm}
If $B \in \mathbb{R}_{p,q,r}^{(2)}$ is a bivector for which $\mu \neq 0$ defines an eigenvalue, and for which $\ta_\wedge(B_\mu)$ exists, then 
\[ \Lambda^{B_\mu} = \Lambda^{B_\mu - \ta_\wedge(B_\mu)}\Lambda^{\ta_\wedge(B_\mu)} = \Lambda^{B_r}\Lambda^{\ta_\wedge(B_\mu)}\ . \]
\end{thm}
\begin{proof}
    All the hard work was done in the lemmas above. It suffices to note that $B_r$ and $\ta_\wedge(B_\mu)$ commute, together with the fact that these bivectors are linearly independent. Lemma \ref{QuadB_product} then does the rest. 
\end{proof}
\noindent
We can now address a small issue raised earlier in the paper, concerning the fact that in defining $B_\mu$ (with $\mu \neq 0$) one has the choice of taking $+\mu$ or $-\mu$. Due to the previous theorem and lemma \ref{out_to_bivector}, the bivector $B_\mu$ (defined with a choice for $\mu$) which thus has eigenvalues $\pm 1$, satisfies
\[ \Lambda^{B_\mu} = \Lambda^{B_\mu - \ta_\wedge(B_\mu)}\Lambda^{\ta_\wedge(B_\mu)}\ \Rightarrow\ B_\mu = \mu\ta_\wedge(B_\mu) + (B_\mu - \mu\ta_\wedge(B_\mu))\ . \]
Because $\ta_\wedge(-B_\mu) = -\ta_\wedge(B_\mu)$, it is immediately clear that replacing $+\mu$ by $-\mu$ in this formula indeed has no impact. Now in order to {\em fully} decompose $B$ into simple bireflections, one may feel like one has to recursively apply the previous theorem, replacing the role of $B$ with $B - \mu\ta_\wedge(B_\mu)$. However, at least if all the eigenvalues are different, one immediately has the following: 
\begin{thm}[Invariant Decomposition]\label{theorem_decomp}
    Suppose $B$ is a regular bivector with $\sigma(B) = \{\pm\mu_1,\ldots,\pm\mu_k\}$ where all eigenvalues are different. One then has that 
    \[ B = \sum_{j = 1}^k b_j = \sum_{j = 1}^k \mu_j \ta_\wedge(B_{\mu_j})\ , \]
    provided all the outer tangents indeed exist.
\end{thm}
\noindent
This is intuitively clear from a symmetry argument, because it does not matter which eigenvalue $\mu_j$ one started the decomposition procedure with. However, a more formal way to see why this works is the following: 
\begin{proof}
    Suppose that $\pm \nu$ is an eigenvalue for $B_r = B - \mu_j \ta_\wedge(B_{\mu_j})$. We then wish to show that 
\[ \nu \ta_\wedge\left(\frac{B_r}{\nu}\right) = \nu \ta_\wedge\left(\frac{B}{\nu}\right)\ , \]
which indeed means that using the very same $B$ with a different eigenvalue $\nu$ is enough. Using the outer trigonometric identity \eqref{out_tan_sum}, and taking into account that for the simple bivector $b_j = \mu_j \ta_\wedge(B_{\mu_j})$ one has that $\si_\wedge(b_j) = b_j$ and $\co_\wedge(b_j) = 1$, we easily find that 
\begin{align*}
    \ta_\wedge\left(\frac{B_r + b_j}{\nu}\right) &= \frac{\ta_\wedge\left(\frac{B_r}{\nu}\right) + \ta_\wedge\left(\frac{b_j}{\nu}\right)}{1 + \ta_\wedge\left(\frac{B_r}{\nu}\right)\left(\frac{b_j}{\nu}\right)}\\
    &= \ta_\wedge\left(\frac{B_r}{\nu}\right)\frac{1 + \ta_\wedge\left(\frac{B_r}{\nu}\right)\left(\frac{b_j}{\nu}\right)}{1 + \ta_\wedge\left(\frac{B_r}{\nu}\right)\left(\frac{b_j}{\nu}\right)} = \ta_\wedge\left(\frac{B_r}{\nu}\right)\ ,
\end{align*}
where we again made critical use of the fact that $\ta_\wedge^2(B) = 1$ for any bivector with eigenvalues $\mu = \pm 1$ (provided $\ta_\wedge(B)$ exists). This shows that there is actually no need to consider the `reduced' bivectors, subtracting the previously obtained simple bireflections $\mu\ta_\wedge(B_\mu)$, it is enough to know the spectrum of $B$ to calculate the simple bireflections using the outer tangent function. 
\end{proof}
\begin{rem}
    Note that the bivectors at the right-hand side in theorem \ref{theorem_decomp} are precisely the simple bivectors $b_j$ appearing in the invariant decomposition from \cite{GSGPS}, i.e. $b_j = \mu_j\ta_\wedge(B_{\mu_j})$ for all $j$. 
\end{rem}
\begin{ex}
To illustrate how this works, consider the (rather trivial, but still insightful) bivector $B = e_{12} + 2e_{34} \in \mathbb{R}_{4,0}$. This bivector has a purely complex spectrum $\sigma(B) = \{\pm i, \pm 2i\}$, so the claim is that
\[ B = i\ta_\wedge\left(\frac{B}{i}\right) + 2i\ta_\wedge\left(\frac{B}{2i}\right)\ . \]
This is indeed true, since 
\begin{align*}
    i\ta_\wedge\left(\frac{B}{i}\right) &= i\frac{-i(e_{12} + 2e_{34})}{1 - 2e_{12}e_{34}} = e_{12}\frac{1 - 2e_{21}e_{34}}{1 - 2e_{12}e_{34}} = e_{12}\\
    2i\ta_\wedge\left(\frac{B}{2i}\right) &= 2i\frac{-\frac{i}{2}(e_{12} + 2e_{34})}{1 - \frac{1}{2}e_{12}e_{34}} = 2e_{34}\frac{1 - \frac{1}{2}e_{12}e_{43}}{1 - \frac{1}{2}e_{12}e_{34}} = 2e_{34}
\end{align*}
Note that the complex numbers $\mu \in \mathbb{C}$ are not instrumental to this story, because $P_{2k}(\mu) = M_\mu \widetilde{M}_\mu$ (see equation (\ref{master_polynomial}) for $k = 2$) is a polynomial in $\lambda = \mu^2$, and similarly $\mu \ta_\wedge(B_\mu)$ is an expression in $\lambda$. Moreover, for the physically relevant cases of rotations, translations and boosts, one will always have that $\lambda \in \mathbb{R}$. 
\end{ex}
\begin{rem}
    The attentive reader may have noticed that we somehow seem to favour the outer tangent function in our argument, despite the fact that the relation $\co_\wedge^2(B_\mu) = \si_\wedge^2(B_\mu)$ can also be read as $\cota_\wedge^2(B_\mu) = 1$. The upshot is that one can choose; both options work. This is encoded in the fact that $\beta_\mu \coloneqq \ta_\wedge(B_\mu)$ is a simple bivector squaring to $1$, so $\beta_\mu^{-1} = \cota_\wedge(B_\mu) = \beta_\mu$. There is however a situation in which one does not really have a choice, and this brings us back to the {\em pseudo-null} bivectors $B$, characterised by eigenvalues $\mu = 0 \in \sigma(B)$. Despite the fact that we cannot properly define $B_\mu$ as a quotient, we do have limit expressions coming to the rescue. Note that this depends on the parity of the parameter $k$ appearing in the effective dimension $2k$. If $k = 2\kappa+1$ is odd, we have that 
\[ \lim_{\mu \rightarrow 0}\mu\ta_\wedge\left(\frac{B}{\mu}\right) = \lim_{\mu \rightarrow 0}\mu\frac{\mu^{2\kappa+1}\si_\wedge(B_\mu)}{\mu^{2\kappa}\co_\wedge(B_\mu)} = \frac{W_k}{W_{k-1}} \]
is well-defined. If $k = 2\kappa$ is even, it suffices to work with the outer cotangent function instead: 
\[ \lim_{\mu \rightarrow 0}\mu\cota_\wedge\left(\frac{B}{\mu}\right) = \lim_{\mu \rightarrow 0}\mu\frac{\mu^{2\kappa}\co_\wedge(B_\mu)}{\mu^{2\kappa-1}\si_\wedge(B_\mu)} = \frac{W_k}{W_{k-1}}\ . \]
Note that this ensures that $W_k$ always appears in the numerator, like it should, since $W_k$ is uninvertable in the case of a pseudo-null bivector. All in all, regular and pseudo-null bivectors can thus be treated on the same footing after all. This observation was already made by the authors of \cite{GSGPS}, but here we managed to reinterpret these limit formulas as special cases in which the standard choice between $\ta_\wedge(B)$ and $\cota_\wedge(B)$ is somehow forced. 
\end{rem}
\noindent
In summary, in order to calculate the invariant decomposition of a bivector $B \in \mathbb{R}_{p,q,r}^{(2)}$ into commuting simple bivectors $b_j$, one first computes the spectrum of the bivector using \cref{master_polynomial}. Subsequently, the commuting simple bivectors are given by $\mu \ta_\wedge(B_\mu)$. This is identical to the closed form solution previously published in \cite{GSGPS}, but explained more succinctly though the lens of outer trigonometry.
As \cref{si=Bco} showed however, this method fails for non-unique eigenvalues, even though the example also indicates that the decomposition of a bivector into commuting simple bivectors still exists.
In the next section we will resolve this problem by alternatively defining $\ta_\wedge(B_\mu)$ not as $\si_\wedge(B_\mu) / \co_\wedge(B_\mu)$ but through the eigenvectors of $B$.

\section{Eigenvectors and outer tangents}\label{sec:eigenouter}
At the end of \cref{si=Bco}, in which we considered an isoclinic bivector $B$, we made an important observation: despite the fact that $\co_\wedge(B)$ was not invertible in the example, we could still define $\ta_\wedge(B)$ because $\si_\wedge(B) = \beta\co_\wedge(B)$ with $\beta$ a simple bivector. This allowed us to cancel out $\co_\wedge(B)$, and to define $\ta_\wedge(B) = \beta$. We will now show that this situation is actually quite general, hence turning the example into a general statement. For that purpose we will assume that a bivector $B$ with non-degenerate eigenvalue $\{+\mu,-\mu\} \subset \sigma(B)$ has two eigenvectors $v_+$ and $v_-$ (we will come back to this non-degeneracy at the end of our argument). Note that these vectors are necessarily null (i.e. they square to zero), as the following result explains: 
\begin{lem}
    If $v_\mu$ is an eigenvector for a bivector $B \in \mathbb{R}_{p,q,r}^{(2)}$, then either $\mu = 0$ or $v_\mu^2 = 0$. 
\end{lem}
\begin{proof}
    Since $v_\mu^2 \in \mathbb{R}$ is scalar, we get that $0 = B \times v_\mu^2 = 2\mu v_\mu^2$, from which the result follows. 
\end{proof}
\noindent
Despite the simplicity of both the statement and its proof, it does have a deep geometrical implication: being null forces the eigenvector to reveal two {\em linearly independent} real vectors which together form a plane in which the action of $B$ happens. This is encoded in the following result: 
\begin{lem}\label{lem:pairing_eigenv}
    If $v_+$ is an eigenvector with eigenvalue $+\mu \in \mathbb{C}_0$ for a bivector $B$ with as many linearly independent eigenvectors as its effective dimension, then there always exists an eigenvector $v_-$ with eigenvalue $-\mu \in \mathbb{C}_0$ such that $v_+ \cdot v_- \neq 0$. 
\end{lem}
\begin{proof}
    First of all we note that any eigenvector $v_\mu$ associated to an eigenvalue $\mu \neq 0$ belongs to the image of $B$, with $v_\mu = \mu^{-1}(B\times v_\mu)$, and as such we get from lemma \ref{lem:W_recursive} that $W_k \wedge v_\mu = 0$. This thus means that $v_\mu$ is a (possibly complex) linear combination of vectors spanning the $2k$-dimensional space for which $W_k$ is the effective pseudoscalar (with $W_k^2 \neq 0$ as $\sigma(B) \subset \mathbb{C}_0$). Next, we note that $B\times (v_{\mu_1} \cdot v_{\mu_2}) = 0$ implies that eigenvectors $v_{\mu_1}$ and $v_{\mu_2}$ corresponding to different eigenvalues $\mu_1 \neq \mu_2$ anti-commute, provided $\mu_1 + \mu_2 \neq 0$. Indeed, only when $\mu_2 = -\mu_1$ the inner product $v_{\mu_1} \cdot v_{\mu_2}$ can be different from zero. This thus means that if {\em all} the eigenvectors were to anti-commute, we would have that 
    \[ v_{+\mu_1} \wedge v_{-\mu_1} \wedge \cdots v_{-\mu_k} \wedge v_{-\mu_k} \propto v_{+\mu_1}v_{-\mu_1}\cdots v_{-\mu_k} \propto W_k\ , \]
    where the constant of proportionality could be complex. This clearly leads to a contradiction, since $W_k^2 \neq 0$, whereas $(v_{+\mu_1}\cdots v_{-\mu_k})^2 = 0$ if all eigenvectors anti-commute (see the lemma above). Hence, there must definitely exist a pair of eigenvectors $v_{+}$ and $v_{-}$ corresponding to eigenvalues $\pm \mu$ such that $v_+ \cdot v_- \neq 0$. Once this is done, the rest follows by an inductive argument: we can now say that $W_k \propto W_{k-1} \wedge (v_+ \wedge v_-)$, where $W_{k-1}$ is proportional to the wedge product of the remaining $(2k - 2)$ eigenvectors. If we now pick an arbitrary (remaining) eigenvector and again assume that it anti-commutes with the other eigenvectors, we would find that $W_{k-1}^2 = 0$ which contradicts the fact that $W_k^2 \neq 0$. We can then repeat this until we have written $W_k$ as an outer product of bireflections $v_{\mu_j} \wedge v_{-\mu_j}$ of `paired eigenvectors' (up to a constant of proportionality). 
\end{proof}
\begin{rem}
    The requirement that $B$ has as many eigenvectors as its effective dimension $2k$ may seem odd, but is really needed here. Suppose for instance that we consider the bivector 
    \[ B = \sum_{a < b}e_{ab} = e_{12} + e_{13} + e_{14} + e_{23} + e_{24} + e_{34} \in \mathbb{R}_{2,2}^{(2)}\ . \]
    One can then easily show that $\sigma(B) = \{\pm 1\}$, where each eigenvalue appears with multiplicity 2, but there are only 2 eigenvectors: $v_{+1} = e_1 - e_4$ and $v_{-1} = e_2 + e_3$. Here we clearly see that $v_+ \cdot v_- = 0$, so a pairing as in the lemma above is not possible here. At the same time, we do have that $\sigma(B) \subset \mathbb{C}_0$ and $W_2^2 \neq 0$, which means that the other requirements of the lemma are still satisfied. The upshot here is that $B$ is special because it exhibits `Jordanesque behaviour', referring to the concept of a Jordan matrix from classical linear algebra. Such bivectors will not be treated in this paper and will be the subject of future research. 
\end{rem}
\begin{rem}
    In some situations, the pairing between $v_{+\mu}$ and $v_{-\mu}$ is easy to describe. For instance, if $\mu \in i\mathbb{R}_0$ is a (non-trivial) purely imaginary eigenvalue, we clearly have that $B\cdot v_{+\mu} = \mu v_{+\mu}$ implies that $B\cdot v_{+\mu}^* = -\mu v_{+\mu}^*$ so that it suffices to put $v_{-\mu} \coloneqq v_{+\mu}^*$ ($*$ hereby stands for the complex conjugation). This happens for instance in $\mathbb{R}_{m,0}$, where bivectors are associated to `classical rotations'. Boosts also give an interesting example: if $\mu \in \mathbb{R}_0$ is real, the eigenvector $v_{+\mu}$ will always be of the form $v_t + v_s$ with $v_t$ a `temporal' vector (i.e. with $v_t^2 > 0$) and $v_s$ a `spatial' vector (i.e. with $v_s^2 < 0$). In that case, it suffices to define $v_{-\mu} = v_t - v_s$, which can for instance be accomplished by conjugating $v_{+\mu}$ with the temporal pseudoscalar $e_P$. 
\end{rem}
\noindent
We will now build further upon the conclusion of lemma \ref{lem:pairing_eigenv} and investigate the bireflection defined by the outer product of an eigenvector $v_{+\mu}$ and its partner $v_{-\mu}$. For that purpose, we introduce the (possibly complex) bivector
\[ \beta_\mu \coloneqq \frac{v_{+\mu} \wedge v_{-\mu}}{v_{+\mu} \cdot v_{-\mu}}\ . \]
The subscript $\mu$ hereby refers to the pair of eigenvalues $\pm \mu \in \mathbb{C}_0$. 
\begin{lem}
    The simple bireflection $\beta_\mu$, with $\beta_\mu^2 = 1$, satisfies the following properties: 
    \begin{itemize}
        \item[(i)] The vectors $v_{\pm \mu}$ are eigenvectors with eigenvalue $\pm 1$. 
        \item[(ii)] One has that $W_k\beta_\mu = \beta_\mu W_k \in \mathbb{R}_{p,q,r}^{(2k-2)}$ (i.e. has grade $2k-2$). 
    \end{itemize}
\end{lem}
\begin{proof}
    First of all, we note that 
    \begin{align*}
        \beta_\mu^2 &= -\frac{1}{(v_{+\mu} \cdot v_{-\mu})^2}(v_{+\mu} \wedge v_{-\mu})(v_{-\mu} \wedge v_{+\mu})\\
        &= \frac{1}{(v_{+\mu} \cdot v_{-\mu})^2}\left(v_{+\mu}v_{-\mu} - v_{+\mu} \cdot v_{-\mu}\right)\left(v_{+\mu} \cdot v_{-\mu} - v_{-\mu}v_{+\mu}\right) = 1\ ,
    \end{align*}
    hereby using the fact that the eigenvectors are null. Next, a simple direct calculation shows that $\beta_\mu \cdot v_{\pm\mu} = \pm v_{\pm\mu}$, hereby using the null property $v_{+\mu}^2 = 0$.
    To prove the second statement of the lemma, we first note that 
    \[ W_{k+1} = 0 \Rightarrow W_{k+1}\cdot v_{+\mu} = W_k \wedge v_{+\mu} = 0 \Rightarrow W_kv_{+\mu} = W_k\cdot v_{+\mu}\ .\]
    This means that $\expval{W_k\beta_\mu}_{2k+2} = \expval{W_k v_{+\mu}v_{-\mu}}_{2k+2} = \expval{(W_k\cdot v_{+\mu})v_{-\mu}}_{2k+2}$. If we can now show that $\expval{W_k \cdot v_{+\mu}}_{2k+1} = 0$, then this indeed proves the statement (the fact that $W_k$ and $\beta_\mu$ commute tells us that the $2k$-graded part $\expval{W_k\beta_\mu}_{2k}$ is trivial, so then only the $(2k-2)$-graded part survives). To see this, we note that 
    \begin{align*}
        B^k \cdot v_{+\mu} = \sum_{j = 1}^k B^{j-1}(B\cdot v_{+\mu})B^{k - 1} = \sum_{j = 1}^k B^{j-1}v_{+\mu}B^{k - j}
    \end{align*}
    has a maximal grade $\ell = 2k-1$ (note how crucial it is here that $v_{+\mu}$ defines an eigenvector). 
\end{proof}
\noindent
The second statement in the lemma above essentially says that $W_k \beta_\mu$ is an effective pseudoscalar (up to scale) in a space of dimension $(2k-2)$, a fact which was used (in a slightly different from) in the proof of lemma \ref{lem:pairing_eigenv}. This suggests introducing the complimentary bivector $B_r \coloneqq B - \mu\beta_\mu$, where the supscript `r' again stands for `the rest' (after subtraction). Our aim is then to show that $\si_\wedge(B_\mu) = \beta_\mu\co_\wedge(B_\mu)$, which implies that $\ta_\wedge(B_\mu) = \beta_\mu$ (provided again the outer cosine is invertible). Writing this out in components, this amounts to proving that 
\[ W_1^\mu + W_3^\mu + W_5^\mu + \ldots = \beta_\mu(1 + W_2^\mu + W_4^\mu + \ldots)\ . \]
Comparing the grades at both sides, this boils down to the following technical result (in which we will suppress the upper index $\mu$, to avoid exceedingly overloaded notations): 
\begin{lem}\label{lem:cos_sin}
    For all (relevant) indices $j$, one has that 
    \[ W_{2j+1} = \expval{\beta_\mu(W_{2j} + W_{2j+2})}_{4j+2} = \beta_\mu \wedge W_{2j} + \beta_\mu \cdot W_{2j+2}\ . \]
\end{lem}
\begin{proof}
    Newton's binomial formula tells us that 
    \begin{align*}
        W_{2j} &= \tfrac{1}{(2j)!}\expval{(\beta_\mu + B_r)^{2j}}_{4j} = \tfrac{1}{(2j)!}\expval{B_r^{2j} + 2j\beta_\mu B_r^{2j-1} + \order{\beta_\mu^2}}_{4j}\ ,
    \end{align*}
    whereby $\order{\beta_\mu^2}$ stands for higher order terms in $\beta_\mu$. The crucial thing to note here is that because $\beta_\mu^2 = 1$, these terms can never contribute to the $4j$-graded part. As a matter of fact, here we have that 
    \[ \expval{\beta_\mu W_{2j}}_{4j+2} = \tfrac{1}{(2j)!}\expval{\beta_\mu (B_r^{2j} + 2j\beta_\mu B_r^{2j-1})}_{4j+2} = \tfrac{1}{(2j)!}\expval{\beta_\mu B_r^{2j}}_{4j+2}\ , \]
    so we did not even use the second term in the expansion. However, for $W_{2j+2}$ it is precisely the second term we will need, since 
    \begin{align*}
        W_{2j+2} &= \tfrac{1}{(2j+2)!}\expval{B_r^{2j+2} + (2j+2)\beta_\mu B_r^{2j+1} + \order{\beta_\mu^2}}_{4j+4}\ .
    \end{align*}
    Multiplying with $\beta_\mu$, we indeed find that only the second term will contribute: 
    \[ \expval{\beta_\mu W_{2j+2}}_{4j+2} = \tfrac{1}{(2j+2)!}\expval{(2j+2)\beta^2_\mu B_r^{2j+1})}_{4j+2} = \tfrac{1}{(2j+1)!}\expval{B_r^{2j+1}}_{4j+2}\ . \]
    If we now add these contributions together, we get
    \begin{align*}
        \expval{\beta_\mu (W_{2j} + W_{2j+2})}_{4j+2} &= \tfrac{1}{(2j+1)!}\expval{B_r^{2j+1} + (2j+1)\beta_\mu B_r^{2j}}_{4j+2}\ ,
    \end{align*}
    and because of the grade selection operator $\expval{\cdots}_{4j+2}$ appearing here, we can add the required lower grade terms to once again apply Newton's binomial formula (in the opposite direction), leading to 
    \[ \expval{\beta_\mu (W_{2j} + W_{2j+2})}_{4j+2} = \tfrac{1}{(2j+1)!}\expval{(B_r + \beta_\mu)^{2j+1}}_{4j+2} = W_{2j+1}\ . \]
    This proves the lemma. 
\end{proof}
\noindent
We can now indeed conclude that $\si_\wedge(B_\mu) = \beta_\mu \co_\wedge(B_\mu)$, and this then implies that $\ta_\wedge(B_\mu) = \beta_\mu$. It is important to point out here that $\si_\wedge(B_\mu) = \beta_\mu \co_\wedge(B_\mu)$ allows us to define $\ta_\wedge(B_\mu)$ {\em without} having to think about the invertibility of the outer sine and/or cosine functions. As a matter of fact, in the case of an isoclinic bivector (see for instance the remark below) one typically has that $\co_\wedge(B_\mu)$ will not be invertible. The upshot is of course that this has no impact on the existence of $\beta_\mu = \ta_\wedge(B_\mu)$. Describing invertible self-reversed elements (such as the outer cosine of a bivector) in full generality is a related, but much harder question which will not be tackled in the present paper. 
\begin{rem}
    Let us then come back to the requirement we stated earlier this section: $\mu = \pm 1$ had to be an eigenvalue with algebraic multiplicity equal to one. Again turning our attention to the bivector $B = e_{12} + e_{34} \in \mathbb{R}_{4,0,0}$, it is clear that one must be careful when there are different eigenvectors for a shared eigenvalue. One could take $v_+ = e_1 + ie_2$ and $v_- = e_3 - ie_4$ here, but this will not work. Not only is $\beta_\mu$ not even real here, but one also observes that the normalisation factor $v_+ \cdot v_- = 0$. But this seemingly annoying fact is actually a blessing in disguise, because the `correct' bireflection to work with is the one obtained by pairing the eigenvectors $v_+$ with eigenvectors $v_-$ for which their scalar part $v_+ \cdot v_-$ is non-trivial. 
\end{rem}
\noindent
We then have the following updated version of theorem \ref{theorem_decomp}: 
\begin{thm}
    Suppose $B$ is a regular bivector which has as many eigenvectors as its effective pseudodimension. One then has that 
    \[ B = \sum_{j = 1}^k \mu_j\ta_\wedge(B_{\mu_j}) = \sum_{j = 1}^k \mu_j \frac{v_{+\mu_j} \wedge v_{-\mu_j}}{v_{+\mu_j} \cdot v_{-\mu_j}}\ , \]
    where in case of repeated eigenvalues one must ensure that $v_{+\mu}$ is paired up with a partner eigenvector $v_{-\mu_j}$ such that the denominator in the second summation above is different from zero. 
\end{thm}
\noindent
The difference with theorem \ref{theorem_decomp} lies in the fact degenerate (i.e. repeated) eigenvalues are now allowed. But as was shown in lemma \ref{lem:pairing_eigenv}, eigenvectors $v_{\pm \mu}$ can always be paired up in such a way that $v_{+\mu_j} \cdot v_{-\mu_j} \neq 0$. 

\section{The (outer) tangent and the Cayley transform}\label{sec:cayley}
In this section we will see how the outer trigonometric functions can appear in the framework of the Cayley transform, imposing relations which mimic group morphism properties. Let us first seek inspiration in the complex plane, where the isomorphism $i\mathbb{R} \cong \mathbb{R}_{2,0}^{(2)}$ allows us to say that
\[ C : \mathbb{R}_{2,0}^{(2)} \rightarrow \textup{Spin}(2) : B = \lambda e_{12} \mapsto C(B) := \frac{1 - B}{1 + B} = \frac{1 - \lambda e_{12}}{1 + \lambda e_{12}} \]
is the Cayley transform (with $\lambda \in \mathbb{R}$). This mapping should be contrasted with the classical exponential map, which maps a bivector $B$ to the spin group element
\[ R = e^B = \exp(B) = \sum_{j = 0}^\infty \frac{B^j}{j!} \in\ \textup{Spin}(p,q,r)\ . \]
 Once the bivector $B$ has been decomposed as $B = b_1 + \cdots + b_k$, it is clear that $R = \exp(B)$ can be written as a product of $k$ rotors $R_j = \exp(b_j)$ which will all mutually commute. This inspires us to look at the Cayley transform in such a way that a bivector $B$, which can be decomposed as the sum of simple and commuting bivectors, is also mapped to a product of commuting rotors. For that purpose we first note that if $b_1$ and $b_2$ are simple and commuting, then one has that
\[ C(b_1)C(b_2) = \frac{1 - b_1}{1 + b_1}\frac{1 - b_2}{1 + b_2} = \frac{1 - \frac{b_1 + b_2}{1 + b_1b_2}}{1 + \frac{b_1 + b_2}{1 + b_1b_2}}\ . \]
This computation suggests defining a (new) binary operation on simple and commuting bivectors, by means of
\[ b_1 \oplus b_2 = \frac{b_1 + b_2}{1 + b_1b_2}\ , \]
such that $C(b_1)C(b_2) = C(b_1 \oplus b_2)$.
Extending this to $k$ simple and commuting elements, one has the following: 
\begin{lem}
    If $x_1,\ldots,x_k$ denote $k$ commuting elements (for instance simple bivectors in a geometric algebra), then one has that 
    \[ x_1 \oplus \ldots \oplus x_k = \frac{\si_\wedge(x_1,\ldots,x_k)}{\co_\wedge(x_1,\ldots,x_k)} = \ta_\wedge(x_1,\ldots,x_k)\ . \]
    The outer functions are hereby defined in terms of the elementary symmetric polynomials, in the sense that 
    \begin{align*}
       \si_\wedge(x_1,\ldots,x_k) &= \sum_{i} e_{2i+1}(x_1,\ldots,x_k) \\
       \co_\wedge(x_1,\ldots,x_k) &= \sum_{i} e_{2i}(x_1,\ldots,x_k)\ .
    \end{align*}
    The summations hereby run over all odd (resp. even) indices $2i+1$ (resp. $2i$) which are smaller than or equal to $k$. 
\end{lem}
\begin{proof}
    This property can easily be proved using induction on the para\-meter $k$. The statement holds for $k = 2$, so let us then focus on 
    \begin{align*}
        x_1 \oplus \ldots \oplus x_k \oplus x_{k+1} &= \frac{\frac{\si_\wedge(x_1,\ldots,x_k)}{\co_\wedge(x_1,\ldots,x_k)} + x_{k+1}}{1 + \frac{x_{k+1}\si_\wedge(x_1,\ldots,x_k)}{\co_\wedge(x_1,\ldots,x_k)}}\\
        &= \frac{\si_\wedge(x_1,\ldots,x_k) + x_{k+1}\co_\wedge(x_1,\ldots,x_k)}{\co_\wedge(x_1,\ldots,x_k) + x_{k+1}\si_\wedge(x_1,\ldots,x_k)}\ .
    \end{align*}
    The lemma then follows from the fact that for all indices $1 \leq p \leq k+1$ one has that 
    \[ e_p(x_1,\ldots,x_k,x_{k+1}) = e_p(x_1,\ldots,x_k) + x_{k+1}e_{p-1}(x_1,\ldots,x_p)\ , \]
    a relation which easily follows from the definition of the elementary symmetric polynomials. 
\end{proof}
\noindent
It is clear that if the commuting variables $x_j$ stand for mutually commuting simple bivectors $b_j$, then these functions can be seen as $\si_\wedge(B)$ and $\co_\wedge(B)$ with $B = b_1 + \ldots + b_k$ (which explains why the very same symbols were used to denote these sums of elementary symmetric polynomials). 
Note that $\ta_\wedge(B) = b_1 \oplus \ldots \oplus b_k$ is not necessarily a bivector, but it is still anti-self-reversed (this hinges upon the fact that the $b_j$ commute) and that $C(b_1 \oplus \ldots \oplus b_k)$ will indeed be a rotor $R$. 
\\
\\
We can now `forget' about the simple commuting bivectors $b_j$ again and work with $B$ instead: the mapping
\[ B \in \mathbb{R}_{p,q,r}^{(2)} \mapsto R := \frac{1 - \ta_\wedge(B)}{1 + \ta_\wedge(B)} \in \textup{Spin}(p,q,r) \]
then provides the generalisation of the Cayley transform, mapping a bivector to a rotor. This has the nice consequence that 
\[ \ta_\wedge(B) = \frac{1 - R}{1 + R}\ \Rightarrow\ B = \arctan_\wedge\left(\frac{1 - R}{1 + R}\right)\ , \]
a formula for $B$ which begs the question whether there is a connection with the `standard' tangent function defined in \cite{GSGPS} as
\[ \ta(B) = \frac{\si(B)}{\co(B)} = \frac{e^B - e^{-B}}{e^B + e^{-B}} = \frac{R - \widetilde{R}}{R + \widetilde{R}}\ , \]
where we have put $R = e^B$. In \cite{GSGPS}, the authors proved that this rotor $R = e^B$ can also be expressed as an outer exponential, with
\begin{equation}\label{rotor_outerform}
   R = \expval{R}_0 \Lambda^{T}\ \ \textup{with}\ \ T = \frac{\expval{R}_2}{\expval{R}_0} \in \mathbb{R}_{p,q,r}^{(2)}\ . 
\end{equation}
This formula does not appear literally in the paper, but it easily follows from the so-called `tangent decomposition' in section 8. In view of formula (\ref{rotor_outerform}), it is then clear that 
\[ \ta(B) = \frac{R - \widetilde{R}}{R + \widetilde{R}} = \frac{\Lambda^T - \Lambda^{-T}}{\Lambda^T + \Lambda^{-T}} = \ta_\wedge(T) = \ta_\wedge\left(\frac{\expval{e^B}_2}{\expval{e^B}_0}\right)\ . \]
\begin{ex}
Consider a rotor $R = \exp(B)$ where $B = b_1 + b_2$. Then working out $\ta(B)$ explicitly, we get
        \begin{align*}
            \ta(b_1 + b_2) &= \frac{\si(b_1)\co(b_2) + \si(b_2)\co(b_1)}{\co(b_1)\co(b_2) + \si(b_1)\si(b_2)} = \frac{\ta(b_1) + \ta(b_2)}{1 + \ta(b_1)\ta(b_2)} = \ta_{\wedge}[\ta(b_1) + \ta(b_2)]\ .
        \end{align*}
\end{ex}

\section{Cayley-Hamilton for bivectors}\label{sec:ch}
In matrix language, the Cayley-Hamilton theorem says that square matrices over a commutative ring satisfy their own characteristic polynomial. To arrive at a GA version of this result, we have to turn an arbitrary bivector $B \in \mathbb{R}_{p,q,r}^{(2)}$ into a \emph{mapping} $f$ on vectors, defined by $f(v) := B \times v$. Defining the repeated action as $f^{a}(v) = f(f^{a-1}(v))$ and adding the identity map $f^0(v) = v$, we will eventually prove the following: 
\begin{thm}\label{CH_bivector}
    For any bivector $B$ with effective pseudoscalar $W_k$, the mapping $f(v) = B \times v$ satisfies its own characteristic polynomial, in the sense that 
\begin{equation}\label{CH_bivector_formula}
    P_{2k}(f) = \sum_{j = 0}^k (-1)^{k - j}\expval{W_j^2}_0 f^{2(k-j)} = 0\ .
\end{equation}
\end{thm}
\noindent
First of all, we can rewrite lemma \ref{lem:W_recursive} in terms of the mapping $f$, since
\[ W_{j+1} \cdot f^a(v) = W_j \wedge \big(B \times f^a(v)\big) = W_j \wedge f^{a+1}(v)\ . \]
This holds for all vectors $v$ and can thus be read as a relation for $f$ (omitting $v$ from the formula). Because the (repeated) action of $f$ on a vector will always be vector-valued, we can rewrite the wedge product as follows: 
\[ W_{j+1} \cdot f^a = W_j \wedge f^{a+1} = W_jf^{a+1} - W_j \cdot f^{a+1}\ . \]
We can now repeatedly use this in a way which resembles our derivation of the (eigenvalue) equation $M_\mu v_\mu = 0$, starting from the effective pseudoscalar: 
\begin{align*}
    W_k v = W_k \cdot f^0(v) &=  W_{k - 1}f^1(v) - W_{k - 1}\cdot f^1(v)\\
    &= W_{k - 1}f^1(v) - \big(W_{k - 2}f^2(v) - W_{k - 2} \cdot f^2(v)\big)\ ,
\end{align*}
and so on (until $f^k$ appears). This leads to the following result: 
\begin{thm}
    For a bivector $B \in \mathbb{R}_{p,q,r}^{(2)}$, the associated mapping $f(v) = B \times v$ satisfies
    \begin{equation}\label{eq:sqrtCH}
        0 = M_{f}(v) \coloneqq f^{k}(v) - B f^{k-1}(v) + W_2 f^{k-2}(v) + \ldots + (-1)^k W_k f^0(v)\ .
    \end{equation}
\end{thm}
\noindent
Note that this result holds for arbitrary vectors $v$, not just for eigenvectors $v_\mu$, but in the special case that $f(v_\mu) = \mu v_\mu$, we recover $M_f(v_\mu) = M_\mu v_\mu = 0$. The formula above has a Cayley-Hamilton flavour to it, but the `coefficients' in front of the mappings $f^a$ are not scalar. Because $M_\mu \widetilde{M}_\mu = P_{2k}(\mu)$ we can interpret $M_f(v) = 0$ as a `square root' of the relation (\ref{CH_bivector_formula}). Without any reference to square roots, we can also note that relation (\ref{CH_bivector_formula}) contains powers of the {\em square} of the mapping $f$. A similar observation was also made in \cite{La} in the context of the group SU$(3)$. As is to be expected, we can also use the previous relation to arrive at a Cayley-Hamilton theorem which {\em does} involve scalar coefficients. This then proves theorem \ref{CH_bivector}. 
\begin{proof}
    Choosing the vector $w = f^a(v)$, the lemma above tells us that 
    \[ f^{k+a}(v) - B f^{k+a-1}(v) + W_2 f^{k+a-2}(v) + \cdots + (-1)^k W_k f^a(v) = 0\ , \]
    a relation which still holds for all $v$. The main idea behind the proof is that we will now repeatedly make use of the relation $M_{f}(w) = 0$ for a suitable $a \in \mathbb{N}$ to get rid of the odd powers of $f$. Writing 
    \[ M_f(w) = 0\ \Leftrightarrow\ f^k(w) = Bf^{k-1}(w) + \cdots + (-1)^{k+1}W_k f^0(w)\ , \]
    it is clear that the second equality still holds if we take the 1-graded part of the right-hand side. This will allow us to discard terms which have higher grades. Putting $a = k-1$, the equation $M_f(w) = 0$ gives: 
    \begin{align*}
        f^{2k}(v) &= \expval{Bf^{2k-1}(v)}_1\\
        &= \expval{B(Bf^{2k-2}(v) - W_2f^{2k-3}(v) + \cdots + (-1)^{k+1}f^{k-1}(v))}_1\\
        &= \expval{B^2}_0 f^{2k-2}(v) - \expval{BW_2f^{2k-3}(v)}_1\ .
    \end{align*}
    The first term at the right-hand side is what it should be, and the second term can be rewritten using $M_f(w) = 0$ for $a = k-2$ and picking up an expression for $Bf^{2k-3}(v)$: 
    \begin{align*}
        \expval{BW_2f^{2k-3}(v)}_1 &= \expval{W_2(f^{2k-2}(v) + W_2f^{2k-4}(v) - W_3f^{2k-5}(v) + \textup{L.O.T.})}_1\\
        &= \expval{W_2^2}_0 f^{2k-4}(v) - \expval{W_2W_3f^{2k-5}(v)}_1\ ,
    \end{align*}
    where the L.O.T. (lower order terms, hereby referring to lower exponents which come with a factor $W_j$ which is too high to contribute to the 1-graded part) could safely be ignored. This argument generalises and allows us to rewrite the second term in each step as a new sum of two terms, using $M_f(w) = 0$ for $a = k-j$. In explicit terms: 
    \begin{align*}
        \expval{W_{j-1}W_jf^{2k-2j+1}(v)}_1 &= \expval{W_j^2}_0 f^{2k-2j}(v) - \expval{W_{j+1}W_jf^{2k-2j-1}(v)}_1\ .
    \end{align*}    
    It then suffices to repeatedly use this formula until $j = k$ (or $a = 0$), in which case the second term will disappear because $W_{k+1} = 0$ is trivial. 
\end{proof}
\noindent
Finally, note that since $f(v) = B \times v$ is a linear map, we can also look at the Cayley-Hamilton theorem as proved by Hestenes and Sobczyk in \cite{HS}: 
\begin{equation}\label{He_So_CH}
        0 = \sum_{j=0}^n (-1)^{n-j} \expval{\partial_{(j)} f_{(j)}}_0 \; f^{(n - j)}(v) \, ,
    \end{equation}
with $n$ the dimension and where $\expval{\partial_{(j)} f_{(j)}}_0$ is the scalar part of the so-called simplicial derivative
\begin{equation*}
        \partial_{(r)} f_{(r)} \coloneqq \frac{1}{r!} (\partial^{a_r} \wedge \cdots \wedge  \partial^{a_1}) f(a_1) \wedge \cdots \wedge f(a_r)\ .
\end{equation*}
Since formula (\ref{He_So_CH}) is again a `scalar CH theorem' (the coefficients in front of powers of $f$ are real numbers), this suggests that this version is related to our theorem \ref{CH_bivector}. To see how this works, we first prove some lemmas. Note that we restrict ourselves to non-degenerate signatures $(p,q)$ in what follows, since we will make use of an orthonormal frame (and its dual) for the effective space in which a bivector $B$ acts. The orthonormal frame $\{e_j\}$ and its dual $\{e^j\}$, where $1 \leq j \leq 2k$ if $W_k$ is the effective pseudoscalar, thus satisfy
\[ e_i \cdot e_j = g_{ij} \qquad \mbox{and} \qquad e^i \cdot e_j = \delta^i_{j}\ . \]
\begin{lem}\label{lem:simplicial}
    For a bivector $B \in \mathbb{R}_{p,q}^{(2)}$, one has that $\partial_{(1)}f_{(1)} = -2B$. 
\end{lem}
\begin{proof}
This follows from direct calculations (or see \cite{La}). 
\end{proof}
\noindent
Next, let us look at the highest-grade part of the simplicial derivative (again for the map $f$ associated to a bivector $B$): 
\begin{lem}\label{lem:simplicial_Wj}
    For all $1 \leq j \leq k$ one has that $\expval{\partial_{(j)}f_{(j)}}_{2j} = (-2)^jW_j$. 
\end{lem}
\begin{proof}
The result follows by direct computation and the lemma above: 
    \begin{align*}
        \expval{\partial_{(j)} f_{(j)}}_{2j} &= \tfrac{1}{j!} \; \partial^{a_j} \wedge \cdots \wedge  \partial^{a_1} \wedge f(a_1) \wedge \cdots \wedge f(a_j) \\
        &= \tfrac{1}{j!} \; (\partial^{a_1} \wedge f(a_1)) \wedge \cdots \wedge (\partial^{a_j} \wedge f(a_j)) \\
        &= \tfrac{1}{j!} \; (- 2 B) \wedge \cdots \wedge (-2 B) = (-2)^j W_j\ ,
    \end{align*}
as was to be shown. 
\end{proof}
\noindent
Before we formulate the final conclusion, we consider an example to illustrate the two main ideas behind the general proof: 
\begin{align*}
    \expval{\partial_{(2)}f_{(2)}}_0 &= \sum_{a < b}\expval{\partial^b \wedge \partial^a f(e_a) \wedge f(e_b)}_0 = \tfrac{1}{4}\sum_{a,b} \expval{\expval{\partial^b\partial^a}_2\expval{f(e_a)f(e_b)}_2}_0\ .
\end{align*}
First of all, note that the summation over $a < b$ was replaced by the full summation over $a$ and $b$, which allows to replace the wedge products by {\em half of} an ordinary GA product. To exclude the contribution coming from $a = b$ we included a grade-2 projection, but using the property $\expval{ABCD}_0 = \expval{BCDA}_0$ this can be omitted. Indeed, moving the derivative $\partial^b$ next to $f(e_b)$ and using lemma \ref{lem:simplicial} will again lead to a product of grade-2 elements: 
\[ \expval{\partial_{(2)}f_{(2)}}_0 = \frac{1}{4}\sum_a\sum_b \expval{\partial^af(e_a)f(e_b)\partial^b}_0 = \expval{(-B)(-\widetilde{B})}_0 = -\expval{W_1^2}_0\ . \]
\begin{thm}
    Let $B \in \mathbb{R}_{p,q}^{(2)}$ be a bivector with associated map $f$, then
    \[ \expval{\partial_{(2j)} f_{(2j)}}_{0} = (-1)^j\expval{W_j^2}_0  \]
    for all $1 \leq j \leq k$, with $k$ the effective dimension of $B$. 
\end{thm}
\begin{proof}
    Like in the example above (the case $j = 1$), we will start by rewriting the summation over indices $a_1 < b_1 < \cdots < a_j < b_j$. We will switch to a full summation over all indices $a_1, \cdots, b_j$ but this comes with a correction factor: 
    \begin{align*}
        \sum_{a_1 < \cdots < b_j}\partial^{b_j} \wedge \cdots \wedge \partial^{a_1} &= \frac{1}{2^j j!}\sum_{a_1}\cdots\sum_{b_j}\expval{(\partial^{b_j}\partial^{a_j})\cdots(\partial^{b_1}\partial^{a_1})}_{2j}\ .
    \end{align*}
    Indeed, there is a compensating factor $2$ per pair of indices $(a_i,b_i)$, and a factor $j!$ to compensate the permutations leading to the same projection on the $(2j)$-graded part. Since $\partial_{(2j)} f_{(2j)}$ is by definition a product of two elements of grade $(2j)$, we can safely perform manipulations until we have again obtained a product of such elements. Using the cyclic property mentioned above, we will move the partial derivatives until they are again next to their partnered function (either from the left or the right), which means that factors $\pm 2B$ will appear. During this process, one may use that $B\partial^i = \partial^iB + [B,\partial^i]$, whereby this second term can safely be ignored since it will lower the grade and hence cannot contribute to the result. Combining these ideas, this means that 
    \begin{align*}
        \expval{\partial_{(2j)}f_{(2j)}}_0 &= \frac{1}{(2^jj!)^2}\sum_{a_1,\cdots,b_j}\expval{\expval{\cdots(\partial^{b_i}\partial^{a_i})\cdots}_{2j}\expval{\cdots(f(a_i)f(b_i))\cdots}_{2j}}_0\\
        &= \frac{1}{(2^jj!)^2}\expval{(-2B)^j(-2\widetilde{B})^j}_0 = (-1)^j\expval{W_j^2}_0\ ,
    \end{align*}
    where we gathered the factors $B$ in such a way that two elements of grade $(2j)$ appeared naturally. 
\end{proof}
\noindent
Using a similar argument, one shows that $\expval{\partial_{(2j+1)}f_{(2j+1)}}_0 = 0$, after which it is clear that the Cayley-Hamilton theorem from Hestenes and Sobczyk is indeed compatible with our \cref{CH_bivector}.\\
\noindent
This leads us to the interesting conclusion that \cref{eq:sqrtCH} is a `square root' of the Cayley-Hamilton theorem, which to the authors best knowledge has no known matrix equivalent. 
We conclude this section with an example of the matrix equivalent to \cref{eq:sqrtCH}, which clearly illustrates that it represents something new.
\begin{ex}
To see the connection with the classical Cayley-Hamilton theorem and matrices, we consider a bivector $B = \sum_{i < j} B_{ij} e_{i}e_j$ in 4 dimensions (Euclidean signature). The map $f(v) = B \times v$ has a matrix representation
    \begin{equation}
        A = \smqty(0 & - B_{12} & - B_{13} & - B_{14} \\
        B_{12} & 0 & - B_{23} & - B_{24} \\
        B_{13} & B_{23} & 0 & - B_{34} \\
        B_{14} & B_{24} & B_{34} & 0 \\) \, .
    \end{equation}
    The characteristic equation of $A$ is given by
    \begin{align*}
        0 = \det(A - \mu \mathbb{I}) &= \mu^{4} + \mu^{2} \left(B_{12}^{2} + B_{13}^{2} + B_{14}^{2} + B_{23}^{2} + B_{24}^{2} + B_{34}^{2}\right) \\
        &\quad + B_{12}^{2} B_{34}^{2} - 2 B_{12} B_{13} B_{24} B_{34} + 2 B_{12} B_{14} B_{23} B_{34} \\
        &\quad + B_{13}^{2} B_{24}^{2} - 2 B_{13} B_{14} B_{23} B_{24} + B_{14}^{2} B_{23}^{2} \\
        &= \mu^4 - \mu^2 \expval{B^2}_0 + \expval{W_2^2}_0\ .
    \end{align*}
    This equation is then identical to $P_{4}(\mu) = 0$, as was to be expected. One may wonder how the matrix realisation of our refined Cayley-Hamilton theorem $M_f(v) = 0$ looks like. Explicitly, we have
    \begin{align}\label{eq:Mf_example}
        0 = M_f(v) &= B \times (B \times v) - B (B \times v) + \tfrac{1}{2} (B \wedge B) v\ .
    \end{align}
This formula clearly contains terms of grade 1, grade 1 and 3, and grade 3. It is therefore impossible that $M_f(v)$ can be written in terms of the matrix $A$ alone, since $A$ maps vectors to vectors. We therefore need {\em another matrix} to represent the geometric product of $\tfrac{1}{2} (B \wedge B)$ with a vector $v$, which maps vectors $\vec{v}$ into trivectors $\vec{\tau}$:
    \begin{equation*}
        T = (- B_{12} B_{34} + B_{13} B_{24} - B_{14} B_{23}) \smqty(0 & 0 & 0 & 1\\0 & 0 & 1 & 0\\0 & 1 & 0 & 0\\1 & 0 & 0 & 0)
    \end{equation*}
    With some effort (on the part of computer algebra software) it can be shown that $B (B \times v)$ has an $(8 \times 4)$ matrix representation $\smqty(A^2 \\ T) \vec{v}$ and therefore has both a vector and a trivector contribution, and we thus conclude that the matrix representation of \ref{eq:Mf_example} is given by
    \begin{equation*}
        \mqty(\vec{u} \\ \vec{\tau}) = \mqty(A^2 \\ \mathbf{0}) \vec{v} - \mqty(A^2 \\ T) \vec{v} + \mqty(\mathbf{0} \\ T) \vec{v} = \mqty(\vec{\mathbf{0}} \\ \vec{\mathbf{0}})\ ,
    \end{equation*}
    where $\vec{u}$ and $\vec{v}$ represent vectors, and $\vec{\tau}$ represent the trivector part.
\end{ex}

\section{Conclusion}
In this paper we have shown that the outer exponential $\Lambda^B$ of a bivector, as an alternative for the classical exponential $e^B$, encodes a lot of information about both the spectrum of $B$ and the invariant decomposition (it captures all the invariants associated to $B$). We have then seen how properties of the outer exponential $\Lambda^B$ relate the simple bivectors $b_1 + \cdots + b_k$ appearing in the invariant decomposition for $B$ to the outer tangent function, via the non-invertibility of the scalar quantity $|\Lambda^B|^2$.
Herein lies the main difference with $e^B$, since $|e^B|^2 = 1$ always implies invertibility. This has further implications than the ones explored in this paper, because one can relate bivectors $B$ with $|\Lambda^B|^2 = 0$ to spinors via the idempotents used to define them. This, together with the problem concerning bivectors with `Jordanesque behaviour' will be treated in a follow-up publication.  

\end{document}